\documentclass[journal,comsoc]{IEEEtran}

\usepackage[T1]{fontenc}
\usepackage[pdftex]{graphicx}
\usepackage{amsmath, amsthm}
\usepackage[cmintegrals]{newtxmath}
\usepackage{breqn}
\usepackage{multirow}
\usepackage{array}
\usepackage{balance}
\usepackage[lofdepth,lotdepth]{subfig}
\usepackage[ruled,vlined]{algorithm2e}
\usepackage{mathtools}

\newtheorem{prop}{Proposition}

\begin{document}
\title{Reconfigurable Intelligent Surfaces for the Connectivity of Autonomous Vehicles}

\author{Y. Ugur~Ozcan,
        Ozgur~Ozdemir
        and~Gunes~Karabulut~Kurt,~\IEEEmembership{Senior Member,~IEEE}
\thanks{The authors are with the Department
of Electronics and Communication Engineering, Istanbul Technical University, 34469 Istanbul, Turkey (e-mail: \{ozcanyi, ozdemiroz3, gkurt\}@itu.edu.tr).}
}

\maketitle

\begin{abstract}
The use of real-time software-controlled reconfigurable intelligent surface (RIS) units is proposed to increase the reliability of vehicle-to-everything (V2X) communications. The optimum placement problem of the RIS units is formulated by considering their sizes and operating modes. The solution of the problem is given, where it is shown that the placement of the RIS depends on the locations of the transmitter and the receiver. The proposed RIS-supported highway deployment can combat the high path loss experienced by the use of higher frequency bands, including the millimeter-wave and the terahertz bands, that are expected to be used in the next-generation wireless networks, enabling the use of the existing base station deployment plans to remain operational, while providing reliable and energy-efficient connectivity for autonomous vehicles.
\end{abstract}


\IEEEpeerreviewmaketitle

\section{Introduction}

\IEEEPARstart{V}{ehicle}-to-everything (V2X) communications need to be supported to reveal to the full potential of the connected autonomous vehicle (CAV) applications since even a packet is important, and its loss may lead to serious consequences. The reliability of a communication link mostly depends on the signal-to-noise ratio (SNR) and the fading characteristics of the corresponding wireless channel. 
The evolution of wireless networks introduces the exploration of higher frequency bands to alleviate the severe spectrum scarcity problem that is encountered in the sub-6 GHz band. The use of millimeter-wave (mmWave) band in the 5G New Radio (NR) is a reality today \cite{3gpp}. Furthermore, the use of terahertz bands is expected in 6G networks \cite{6g}, placing further constraints on the reliability requirements as a higher path loss is experienced in higher bands, leading to possibly low SNR values. One solution to address this problem is to increase the transmit power, yet this severely reduces the energy efficiency of the network. Another solution is to reduce the distance between the transmitters and the receivers; however, this solution is also costly from CAPEX and OPEX perspectives as it requires a new base station (BS) deployment plan.

In this study, the use of real-time software-controlled reconfigurable intelligent surfaces (RISs) is proposed to increase reliability for V2X communications. Being able to increase the received signal power and eliminate the fading effect simultaneously, RIS units can be considered as a candidate for the reliable operations of CAVs.  Consisting of meta-material pixels, RIS units can control the reflection angle by altering the reflection coefficients and phase shifts of these pixels independently from each other, in real-time via software.

Energy-efficient beamforming is pursued by the use of RISs in \cite{uc} and \cite{beam}. In \cite{sec} and \cite{sec2}, physical layer security systems are proposed by deploying a RIS. There are also studies focusing on the physical aspects of the RIS, which have crucial effects on the identification of the reflection behaviors, as in \cite{ozgecan} and \cite{nearfar}. However, their roadside use for reliable V2X systems to improve the average SNR, as shown in Fig. \ref{tam}, is not yet considered in the literature. 
\begin{figure}[t]
\centering
\includegraphics[width=\columnwidth]{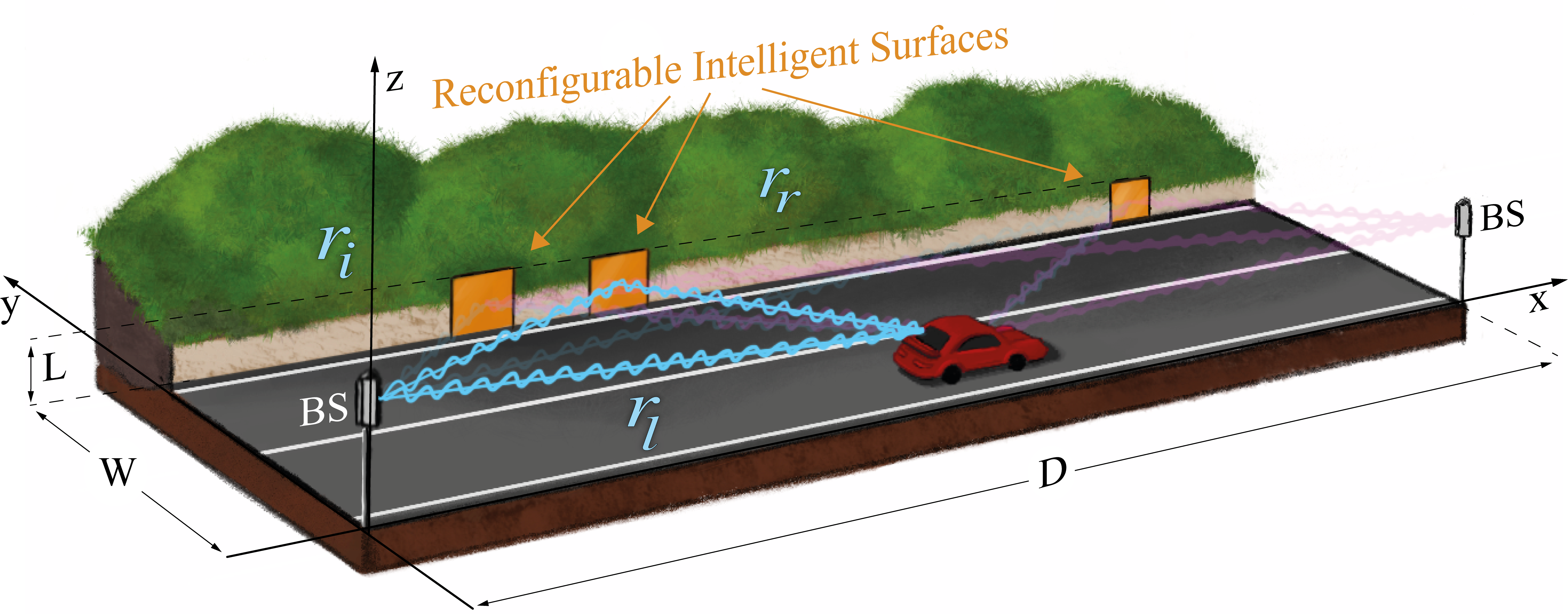}
\captionsetup{justification=centering}
\caption{The roadside deployment scenario for multiple RIS units targeting reliably connected autonomous vehicles.}
\label{tam}
\end{figure}
In this study, we study the optimum placement of the RIS units by considering their sizes and the operating mode. The beamforming and the focusing modes are considered.
The contributions  can be listed as: 
\begin{itemize}
  \item In order to determine the optimum coordinates for the RIS units, an optimization problem has been formulated and solved.
  \item Extensive numerical results are provided to quantify the impact of the proposed solution, which depends on the locations of the transmitter and the receiver. A comparison with the intuitive equidistant placement is given. Results show that significant SNR gains can be obtained by the optimal deployment of the available RIS units.
\end{itemize}
 
 As the use of the mmWave and the THz band in the wireless networks are becoming a reality, our proposed solution can deem the existing coverage plans for sub-6 GHz band inapplicable in sparsely populated highways, where the distance between BSs is determined according to the coverage constraints with the goal of minimizing the infrastructure-related and the associated operational costs. 
 
 The rest of our paper is organized as follows.  In Section II, necessary prior information is given. In Section III, the system configuration for rural highways is described and the optimization problem is presented. Numerical results are presented in Section IV. Conclusions are given in Section V.

\section{Propagation through RIS}
Two main propagation modeling approaches are discussed in the literature for RIS-aided communication systems. In the first approach, RIS units are considered as specular reflectors \cite{ertugrul},  while in the second approach they are considered as diffuse scatterers \cite{ozgecan}.  Both models are studied below.

We consider the downlink of a cell-free transmission scenario, where $K$  BSs and a vehicle is communicating. The total received line-of-sight (LoS) signal, $S_{l}\left(t;r_l\right)$, coming from all BSs for a unit powered message ($u(t)$) is modeled as \cite{goldsmith} 
\begin{equation}
\begin{split}
\label{LOS}
&{S}_{l}\left(t;r_l\right) = \left\{\frac{\lambda}{4\pi} \sqrt{G_B G_V}\left[ \sum_{k=1}^{K}\frac{u(t) e^{-j2\pi (r_{l}^k/\lambda)}}{r_{l}^k}\right]\right\},
 \end{split}
\end{equation} 
where  $r_l^k$ is the distance between the vehicle and the $k^{\textrm{th}}$ transmitter. $G_B$ and $G_V$ represent the antenna gains of the BSs and the vehicle. The signal is assumed to be narrowband relative to the delay spread, hence, time differences amongst different paths from $K$ BSs have been omitted. 

In the specular model, RIS units behave like perfect mirrors which reflect every incident wave according to the Snell's reflection law, i.e., the beam curvature of the reflected wave is the same of the incident beam curvature \cite{balanis}.  Therefore, the incident wave is not affected from an extra path loss. The reflected wave's path loss is equal to the same path loss of a wave, which has the same travelling distance with the reflected one. Then, the baseband received signal from a specular reflector, $S_{s}\left(t;r_i,{r_r}\right)$, can be modeled as
\begin{equation}
\begin{split}
\label{specular}
{S}_{s}&\left(t;r_i,{r_r}\right) = \left\{\frac{\lambda}{4\pi}\sqrt{G_B G_V} \cdot \right.\\ &\left.\left[ \sum_{k=1}^{K}\sum_{n=1}^{N}\frac{\sqrt{G_I^{n}(\mathbf{\hat{r}}_i^{kn})G_I^{n}(-\mathbf{\hat{r}}_r^{n})}
u(t)e^{-j2\pi ((r_{i}^{kn}+r_{r}^{kn})/\lambda)}}{r_{i}^{kn}+r_{r}^{n}}\right]\right\}.
\end{split}
\end{equation}
In the expression above, the received signal is coming from $K$ transmitters via $N$  RIS units, while $r_{i}^{kn}$ and $r_{r}^{n}$ denote the incident and the reflected path distances for the signals coming from the $k^{th}$ transmitter via $n^{th}$ RIS, respectively.  $G_I^{n}(-\mathbf{\hat{r}}_{i}^{kn})$ and $G_I^n(\mathbf{\hat{r}}_{r}^{n})$ are the gains in the incident direction, $\mathbf{\hat{r}}_{i}^{kn}$, and reflection direction, $\mathbf{\hat{r}}_{r}^{n}$,  of the $n^{th}$ RIS, respectively. 

In the diffuse scatterer model, RIS units scatter the incident wave in all directions instead of reflecting it by conserving its original beam curvature as in the specular model. 
Therefore not all the transmitted signal can be caught by the receiving antenna and this results in extra power loss on the receiver side. The total received signal from a diffuse scatterer can be written as \cite{nearfar}  
\begin{equation}
\begin{split}
\label{diffuse}
{S}_{d}&\left(t;r_i,{r_r}\right) = \left\{\frac{A}{4\pi}\sqrt{G_B G_V} \cdot \right.\\ &\left.\left[ \sum_{k=1}^{K}\sum_{n=1}^{N}\frac{\sqrt{G_I^{n}(-\mathbf{\hat{r}}_i^{kn})G_I^{n}(\mathbf{\hat{r}}_r^{n})}
u(t)e^{-j2\pi ((r_{i}^{kn}+r_{r}^{kn})/\lambda)}}{r_{i}^{kn}r_{r}^{n}}\right]\right\}
\end{split}
\end{equation}
where  $A$ is surface area of a RIS unit. 
Note that denominator takes a multiplicative form in this case to model extra power loss  which is more severe than (\ref{specular}). 
The physical interpretation of this situation stems from the fact that the size of the RIS units determines the effective aperture of the RIS which is equal to $ A_{e}  = \frac{\lambda^2}{4\pi} G$. The effective aperture is proportional to the antenna gain $G$ in that direction which is equal to $G = \epsilon \mathcal{D}$, where $\epsilon$ is the efficiency and the $\mathcal{D}$ is the directivity. When the size of a RIS unit is much smaller than the wavelength and the distances between the antennas, it can be considered as a point source that radiates omnidirectionally and its directivity is equal to one.  On the other hand, as the size of the RIS unit gets larger, the directivity increases and the beam gets narrower.  Thus, the specular reflection implies very high directivity which requires an infinitely large RIS, as stated in \cite{balanis}. 

The key point in propagation modeling is to select the most suitable model for the given circumstances. The specular reflection assumption is deployed for ground and building reflections in communication system designs, even though they are not infinitely large \cite{goldsmith}. On the other hand, if the receiver gets closer to the RIS, it can catch more scattered beams since these beams are propagating by expanding. By getting closer to RIS, the same amount of the waves reflected via the specular reflectors can be achieved via the diffuse scatterers as shown in experimentally in \cite{experimental}. This leads to a strong relationship between the RIS size and RIS distances from the transmitter/receivers antennas. Therefore, a critical RIS size in which both models give the same path loss can be defined. To this aim, by equating the two signal levels given in (\ref{specular}) and (\ref{diffuse}), the critical size $A_{min}$ can be determined as follows \cite{nearfar} 
\begin{equation}
\label{th}
\begin{split}
  A_{min}^{kn} = \left[(r_{i}^{kn}
  r_{r}^{n})/(r_{i}^{kn}+r_{r}^{n})\right ]\cdot \lambda .
\end{split}{}
\end{equation}
Here, $A_{min}$ stands for the minimum area for a RIS has the same path loss for each assumption. Having an area larger than $A_{min}$, RIS can act as a specular reflector for these specific $r_{i}$ and $r_{r}$ by adjusting the phases to sustain the same beam curvature. In this case, the path loss equals to the LoS path's with the traveling distance $r_i+r_r$.

It is important to note that having a larger area than $A_{min}$, the specular RIS performance cannot be improved further. This situation can be related with  {\it{beamforming}} mode of the RIS \cite{nearfar}. RIS here steers the wave to the direction of the receiver without changing the beam curvature of the wave. In other words, the RIS reflects the wave to the direction of the receiver, not to the receiver location. And this case is very useful when the exact location of the receiver is unknown or difficult to determine. Yet, this does not correspond to the full potential of the RIS which is capable of focusing all the incident waves to a point. By focusing all the waves, the RIS can prevent the power loss and it is called {\it{focusing}} mode  \cite{nearfar}. When RIS and the antennas are far away from each other, these two operating modes do not give significantly different results in terms of path loss. However, when the RIS and the antennas are closer to each other, the RIS in focusing mode can perform better than the specular reflection. 

\section{System Model}

As a V2X communication system model, it is considered that the autonomous vehicle traveling on the highway is connected to the network via cell-free mobile communication, and several numbers of RIS units have been placed perpendicularly to the road, as shown in Fig. \ref{tam}. Using this model, the effect of beamforming and focusing modes of RIS are examined to improve autonomous vehicle connectivity along a rural highway.  In order to analyze the system, the total power consisting of the power of the LoS component and the power of the RIS-directed links received by the vehicle on the highway between two cell-free BSs is calculated as follow:
\begin{equation}
\begin{split}
\label{powbf}
{P}_{total} \left(r_l,r_i,r_r\right) = \left| S_{l}+(1-\beta) S_{d}+\beta S_{s}\right|^2,
 \end{split}
\end{equation}
where $\beta$ is a binary element which stands for a RIS mode activation multiplier such as  $\beta$ is chosen as $0$ for the focusing mode. And for the beamforming mode, it can take $0$ or $1$ depending on the size of the surface area of RIS. If the RIS area $A$ is larger than  $A_{min}$ in (\ref{th}), $\beta$ takes the value of $1$, otherwise it is $0$. Note that when $\beta$ equals $1$, the received power of the beamforming mode will be the same as in specular reflections case \cite{nearfar}. For simplicity purposes, the phase differences among the LoS waves and RIS-directed waves are eliminated by choosing $e^{-j2\pi (r/\lambda)}=1$.

In this study,  the locations of the RIS units have been optimized to maximize the average received power of the vehicle along the road.  To this aim,  the objective function and the constraints have been set as
\begin{equation}
\begin{aligned}
& \underset{\mathbf{x_I}}{\arg\max}\
& & {P}_{avg}= \frac{1}{J}\sum_{j=1}^{J}  \left| {S}_{l}+(1-\beta) S_{d}+\beta S_{s}\right|^2\\
& \text{subject to}
& & x_I^n\leq D,\ \    x_I^{n}\in \mathbb{Z}^+, \ \ |x_I^{n} - x_I^{n-1}| \ \geq \Delta    \ \ \ &\forall{n} \in \mathbb{Z}^+ \\
&&& {x_V}_j = j\delta_{x_V} \leq D,\ \ \   {x_V}_j \in \mathbb{Z}^+\ \ \ &\forall{j}\in \mathbb{Z}^+\\ 
&&&\text{Beamforming case;}\\
&&&   \ \ \  A_{min}^{kn} + M\beta \geq A \geq A_{min}^{kn} - M(1-\beta)\\
&&&\ \ \  \ \beta \in \{0,1\} \\
&&&\text{Focusing case;}\\
&&&\ \ \  \ \beta = 0.
\end{aligned}
\label{opti2}
\end{equation}
  $D$ is the distance between two BSs and $M$ is a very large number used as big M-constraint in the optimization routine. $\Delta$ is the proximity constraint of RIS locations which is set by the system designer by considering $\Delta > A$ relation. Here, the center points of $n^{th}$ RIS, $k^{th}$ BS and vehicle are defined as $R_I^n = (x_I^n,y_I^n,z_I^n)$, $R_B^k =(x_B^k,y_B^k,z_B^k)$, $R_V = (x_V, y_V,z_V)$, respectively. Accordingly, the distances are defined as $r_{i}^{kn} = \|R_{B}^k - R_{I}^n\|$,  $r_{r}^{n} = \|R_{V} - R_{I}^n\|$ and $r_{l}^{k} = \|R_B^k - R_V\|$. 

The optimization routine outcome  $\mathbf{x_I}$ gives the optimum $x$ coordinates of $N$   RIS units. In the algorithm, the beamforming or the focusing mode is chosen initially depending on the information of vehicle location  and then each mode is optimized individually in the routine.
First, the following propositions are stated that can be used to solve the proposed optimization problem.

\begin{prop}
\label{prop1}
The objective function for focusing mode can be reduced to the following function for each $n$ 
\begin{equation}
\label{prop1eq}
\mathbf{x_I^*} = \underset{x_I^n}{\arg\min}
(r_{r}^{n}\displaystyle\prod_{k=1}^{K} r_{i}^{kn}).
\end{equation}
\end{prop}
\begin{proof}  In the optimization problem of (\ref{opti2}), only $r_{i}$ and $r_{r}$ variables are the functions of the optimization parameter $\mathbf{x_I}$. When $\beta=0$ is inserted into the objective function and the constant terms are excluded, it can be reduced to 
\begin{equation}
\begin{split}
\label{d}
& \underset{\mathbf{x_I}}{\arg\max} 
\sum_{j=1}^{J} \left| \sum_{n=1}^{N}\sum_{k=1}^{K}\frac{
1}{r_{i}^{kn}r_{r}^{n}}\right|^2.
 \end{split}
\end{equation}
Since  ${r_{i}^{kn}}r_{r}^{n} \geq 0 $, the maximization problem can be considered as
\begin{equation}
\begin{split}
\label{efiki}
& \underset{\mathbf{x_I}}{\arg\max}
\sum_{j=1}^{J}\sum_{n=1}^{N}\sum_{k=1}^{K}\frac{ 
1}{r_{i}^{kn}r_{r}^{n}},
 \end{split}
\end{equation}
where the term inside the summation is 
\begin{equation}
\label{numeric}
\frac{ 
1}{r_{i}^{kn}r_{r}^{n}}= \frac{1}{\sqrt{
\begin{aligned}
[(x_B^k- x_I^n)^2+(y_B^k- y_I^n)^2+(z_B^k- z_I^n)^2]\\
[(x_{V}- x_I^n)^2+(y_V- y_I^n)^2+(z_V- z_I^n)^2]
\end{aligned}}}.
\end{equation}

It is clearly seen that the maximum value of (\ref{numeric}) is achieved at either  $ x_I^n = x_B^k$ or $ x_I^n = x_{V}$.  Since  $ x_I^n = x_{V}$ is not physically possible, the maximum value of (\ref{efiki})  is obtained at $ x_I^n = x_B^k$ for all $n$. It corresponds to the nearest points to the BSs $ x_I^n = x_B^k$. Thus, the dependence on the vehicle and the number of RIS can be eliminated from the summation:
\begin{equation}
\begin{split}
\label{eq}
& \underset{\mathbf{x_I}}{\max}\  \sum_{j=1}^{J}\sum_{n=1}^{N}\sum_{k=1}^{K}\frac{ 
1}{r_{i}^{kn}r_{r}^{n}} = \sum_{j=1}^{J}\sum_{n=1}^{N}\underset{\mathbf{x_I}}{\max}\sum_{k=1}^{K} \frac{1}{r_{i}^{kn}r_{r}^{n}}.
 \end{split}
\end{equation}
The resulting optimization problem now can be solved analytically and the solution is equal to physically possible minimum value of $r_{i}^{kn}$. Proposition 1 is thus proved. 
\end{proof}

\begin{prop}
\label{prop2}
 
 In the beamforming mode, the optimum solution of $x_I^n$ can be achieved for each $n$ as

\begin{equation}
\label{prop2eq}
\mathbf{x_I^*} = \underset{{x_I^n}}{\arg\min}
\displaystyle\prod_{k=1}^{K} \left|({r_{i}^{kn}+r_{r}^{n}}) - \frac{(r_{i}^{kn}r_{r}^{n})\cdot \lambda}{A}\right|
\end{equation}
under the assumption of $r_{r}^{n} \gg r_{i}^{kn}$.
\end{prop}

\begin{proof} 

Since  ${r_{i}^{kn}}$ and $r_{r}^{n}$ are the only ($\mathbf{x_I}$) dependent variables the optimization problem (\ref{opti2}) is reduced to the following form for beamforming mode:
\begin{equation}
\begin{split}
\label{dd}
& \underset{\mathbf{x_I}}{\arg\max}
\sum_{j=1}^{J} \left| \sum_{n=1}^{N}\sum_{k=1}^{K}\frac{(1-\beta)}{(r_{i}^{kn}r_{r}^{n})} + \frac{\beta}{(r_{i}^{kn}+r_{r}^{n}) }\right|^2.
 \end{split}
\end{equation}
It is seen that the dominant term in (\ref{dd}) belongs to the specular part, thus, the diffuse part can be eliminated with  $\beta = 1$. Then the objective function can be written as
\begin{equation}
\begin{split}
\label{d2}
& \underset{\mathbf{x_I}}{\arg\max}
\sum_{j=1}^{J} \left| \sum_{n=1}^{N}\sum_{k=1}^{K}\frac{ 
1}{r_{i}^{kn}+r_{r}^{n} }\right|^2.
 \end{split}
\end{equation}
Since ${r_{i}^{kn}}+r_{r}^{n} \geq 0 $, this maximization problem is equivalent to the following:
\begin{equation}
\begin{split}
\label{spec}
& \underset{\mathbf{x_I}}{\arg\max}
\sum_{j=1}^{J}\sum_{n=1}^{N}\sum_{k=1}^{K}\frac{ 
1}{r_{i}^{kn}+r_{r}^{n}}.
 \end{split}
\end{equation}
The dependence on the vehicle and the number of RIS are  eliminated from the summation as explained in proof of Proposition \ref{prop1}:  
\begin{equation}
\begin{split}
\label{eq2}
& \underset{\mathbf{x_I}}{\max}\  \sum_{j=1}^{J}\sum_{n=1}^{N}\sum_{k=1}^{K}\frac{ 
1}{r_{i}^{kn}+r_{r}^{n}} = \sum_{j=1}^{J}\sum_{n=1}^{N}\underset{\mathbf{x_I}}{\max}\sum_{k=1}^{K} \frac{1}{r_{i}^{kn}+r_{r}^{n}}.
 \end{split}
\end{equation}Moreover for every $n$, this function can be converted into 
\begin{equation}
\label{bf4}
\underset{{x_I^n}}{\arg\min}
\displaystyle\prod_{k=1}^{K} (r_{i}^{kn}+r_{r}^{n}).
\end{equation}

Notice that for specular mode,  this minimization problem subjects to the constraint specified by (4). $A_{min}$, in general,  depends on the vehicle location, however,  with the assumption of $r_{r} \gg r_{i}$, it will be fixed with respect to the vehicle location. 
 By rearranging (4) to $A \geq A_{min}^{kn} = \left[(r_{i}^{kn}r_{r}^{n})/(r_{i}^{kn}+r_{r}^{n})\right ]\cdot \lambda$ and inserting it into the (\ref{bf4}) as a constraint $(r_{i}^{kn}+r_{r}^{n})-[(r_{i}^{kn} r_{r}^{n})\cdot\lambda/A] \geq 0$, (\ref{prop2eq}) can be obtained.
 \end{proof}

\begin{algorithm}[t]
\SetAlgoLined
Select the mode as \textit{Beamforming} or \textit{Focusing}; \\
\textbf{Initialize}  $n=1$, $\Omega : x \in [0, D]$, $\Delta$\;
\Repeat{$n > N$}{
\eIf{Focusing Mode is selected}{
Find the optimum points ($\mathbf{x_{I}^{*}}$) by solving (\ref{prop1eq}).\\
}{
 Find the optimum points ($\mathbf{x_{I}^{*}}$) by solving (\ref{prop2eq}).
 }
 \eIf{$\mathbf{x_{I}^{*}}$ has more than one element}
 {$x_I^n = min(\mathbf{x_{I}^{*}})$}
 {$x_I^n = x_{I}^{*}$}
 $\Omega : \Omega \ \backslash \ \{ |x-x^n|<\Delta \}$\;
 $n = n+1$\;
 }
\caption{Optimum RIS positioning}
\label{algo}
\end{algorithm}
The optimum locations of RIS units can be determined using Propositions 1 and 2 in the proposed optimization algorithm given in Algorithm \ref{algo}. This algorithm recursively selects the optimum positions for the RIS units, by starting from the location of the first RIS unit, and then by using this location, it determines the location of the second RIS unit, and so on. As stated,  the optimum locations are expected to be different in the focusing and the beamforming modes. In the focusing mode, the optimum positions are the nearest points to the BSs, and the performance decreases by getting away from the BSs. On the other hand, for beamforming mode, the optimum location is determined by (\ref{th}) for the first RIS. Then the following RIS should be placed by getting closer to the BSs. In locating two adjacent RIS units, the distance between them is set as the pre-defined $\Delta$ value.

\section{Numerical Results}

\begin{figure}[t]
\centering
\includegraphics[width=\columnwidth]{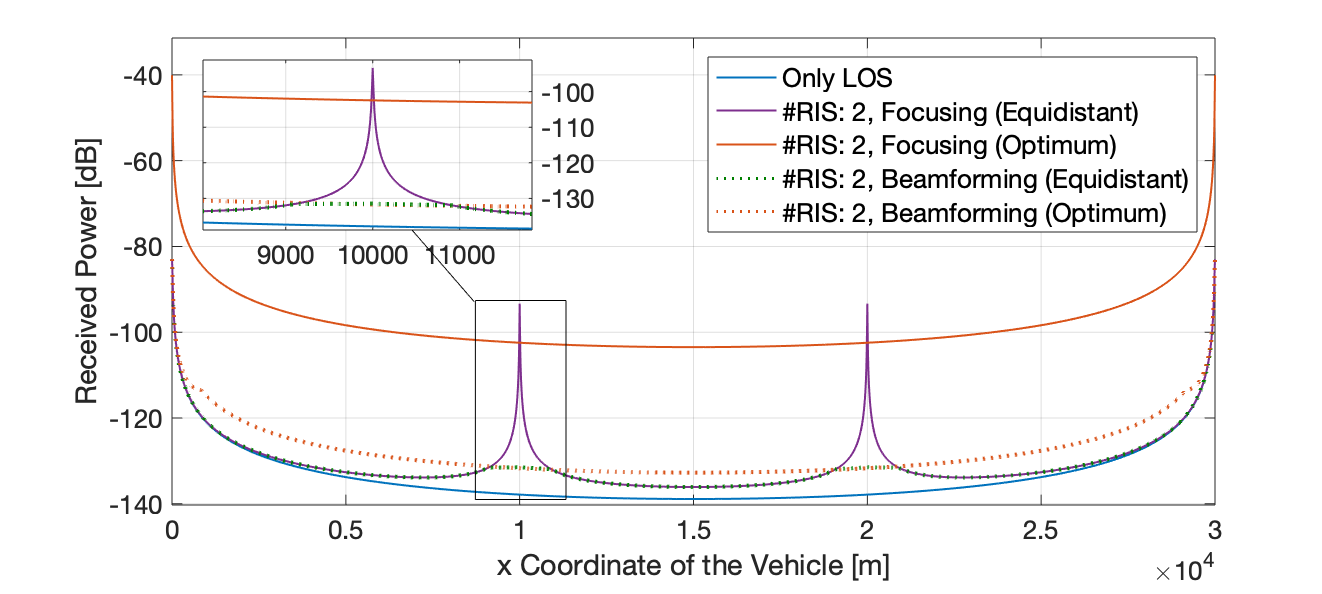}
\captionsetup{justification=centering}
\caption{The comparison of optimized and equidistant RIS positioning effects on received power value for 2 RIS units.}
\label{split1}
\end{figure}
\begin{figure}[t]
\centering
\includegraphics[width=\columnwidth]{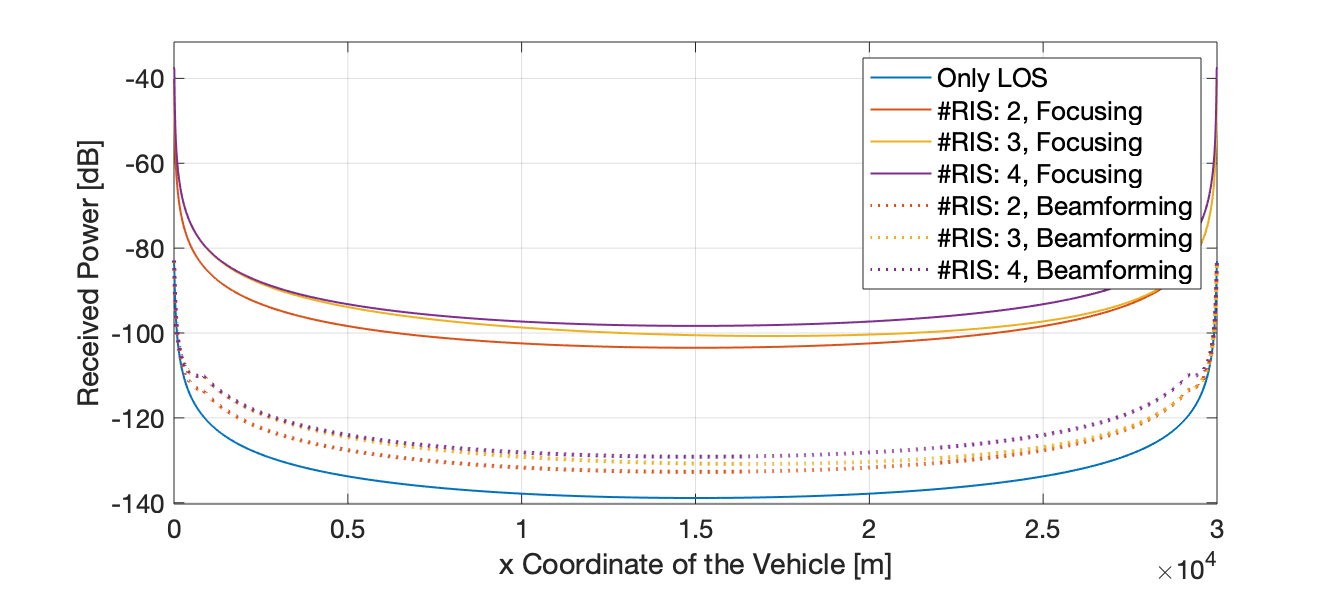}
\captionsetup{justification=centering}
\caption{Received power values for 2, 3 and 4 optimally located RIS units}
\label{split2}
\end{figure}

 \begin{figure*}[t]
  \centering
  \subfloat[Only LOS]{\includegraphics[width=0.3\textwidth]{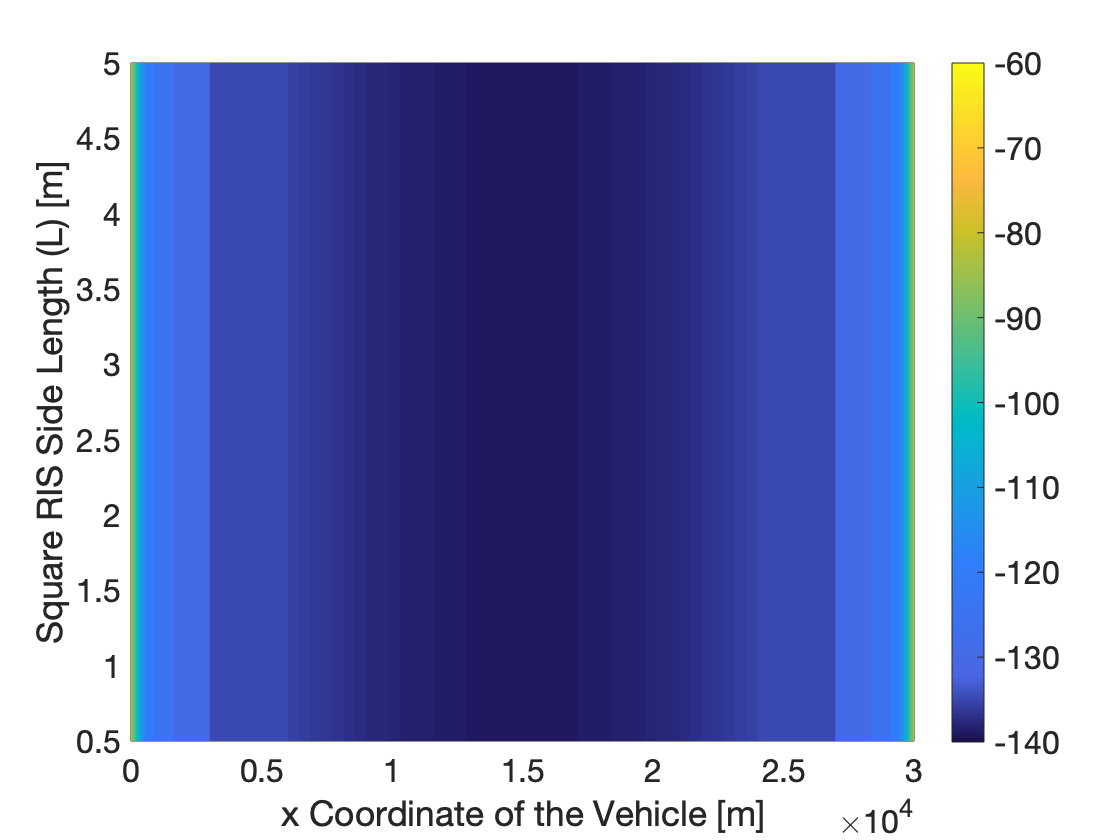}\label{fig:f1}}
  \subfloat[RIS Aided Beamforming]{\includegraphics[width=0.3\textwidth]{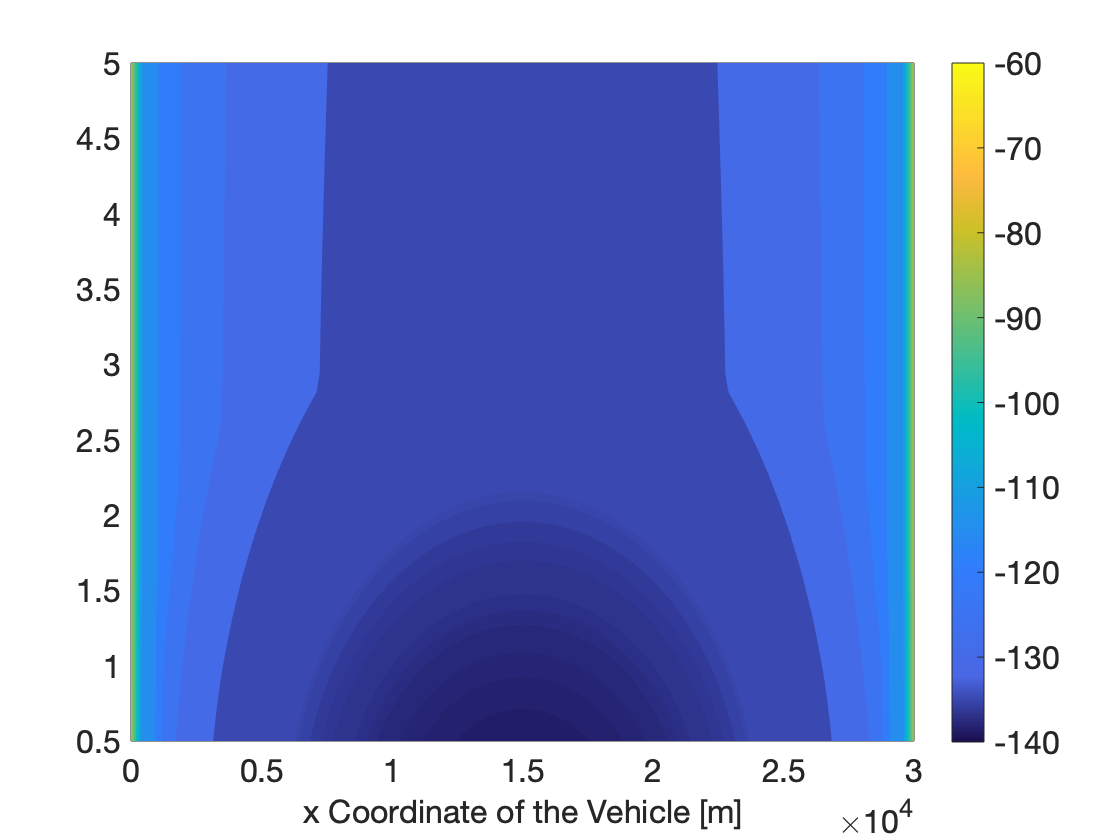}\label{fig:f2}}
  \subfloat[RIS Aided Focusing]{\includegraphics[width=0.3\textwidth]{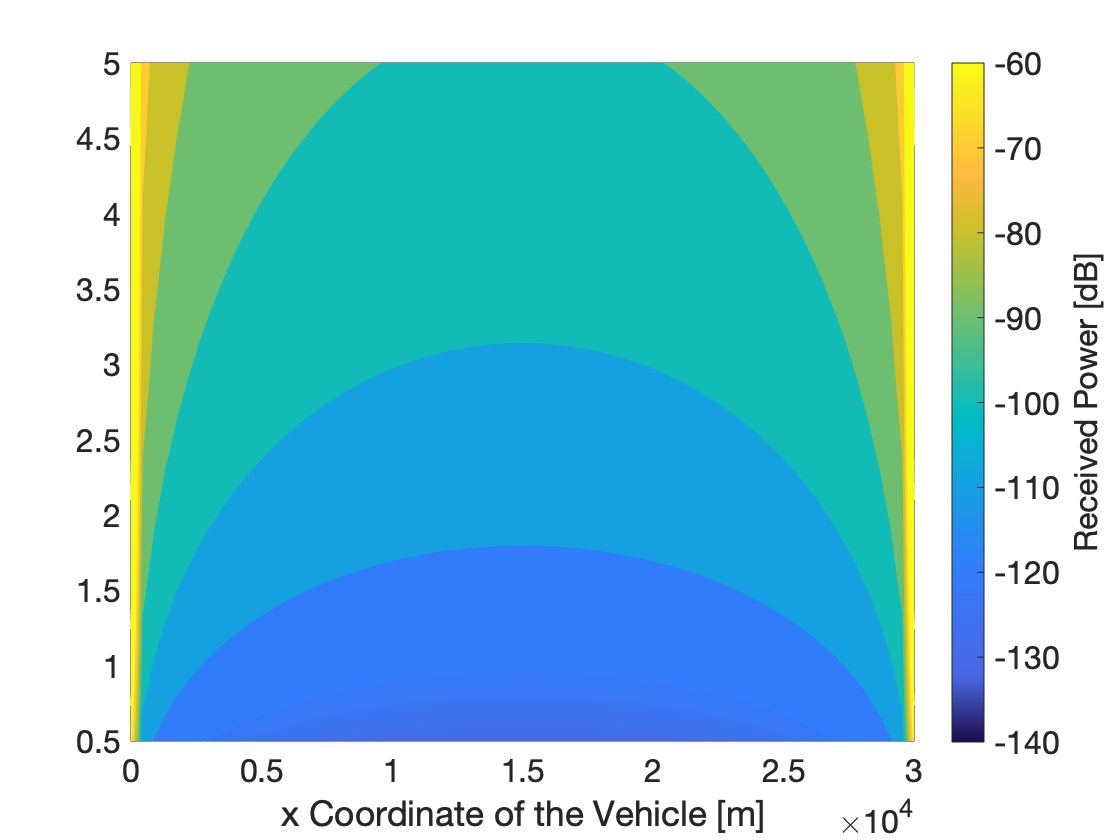}\label{fig:f3}}
  \caption{Optimum received power values depending on RIS sizes in different modes.}
  \label{length}
\end{figure*}

In all simulations, the heights of the RIS units, BSs and the vehicle have been selected as $z_{I} = 12m$, $z_{B} = L/2 = 3/2 m$ and $z_{V}=1m$, respectively. Here, $L$ is the side length of a square RIS unit. Two BSs ($K=2$) that are separated with a distance of $D=30km$ are considered. The surface areas of all RIS units ($A$) and the highway width ($W$) has been selected as $9 m^2$ and $10m$, respectively. Moreover, $y_{I} = -W/2$, $y_{B} = W/2 $ and $y_{V} = 0$ has been used. The wavelength is chosen as $\lambda = 10.7 mm$ corresponding to 28 GHz carrier frequency for NR, and all the antenna gains have selected as $G=1$. The minimum distance between two RIS center points has been set as $\Delta =10m$.

First, the received power levels of the system are compared for two  RIS units which are placed along the road using the optimization routine, specifically for the beamforming and focusing modes. The optimum locations $\mathbf{x_I}$ are found using Algorithm \ref{algo}  as $\{0, 30000\}$ meters for focusing mode and as $\{839, 29161\}$ meters for beamforming mode. The received power levels obtained for two RIS  placed at these locations are presented in Fig. \ref{split1}. For comparison purpose, the power levels of two equidistantly placed RIS units along the road (i.e. $x=\{10000,20000\}$ meters) and only LoS as well, are plotted in Fig. \ref{split1}.
 As seen, the use of the RIS units improves the power levels within the small portion of the road. On the other hand, the optimized placement obtained by the proposed algorithm gives a continuous improvement all along the road. It is also observed from Fig. \ref{split1} that the focusing mode outperforms the beamforming mode as expected. Note that the RIS  units in the focusing mode are only applicable when the exact location of the vehicle is known, thus, all the power can be transmitted to the vehicle. However, it is a quite difficult task to monitor the position of the vehicle all the time. On the other hand, the RIS in beamforming mode requires only the direction information which is more attainable than the exact location, though its performance is worse than the focusing mode.

In order to examine the effect of the number of RIS on the received power,  3 and 4 RIS units placement are considered. For 3 RIS,  the locations are obtained as $\{0, 10, 30000\}$ meters for the focusing mode and as $\{829, 839, 29161\}$ meters for the beamforming mode. And 4 RIS units,  the locations are obtained as  $\{0, 10, 29990, 30000\}$ meters for the focusing mode and $\{829, 839, 29161, 29171\}$ meters for the beamforming mode. As seen in Fig. \ref{split2}, the increase in the number of RIS results in the increase in the power level for both modes. Here, it can be seen that optimizing the locations using 2 and 4 RIS units can increase the received power in focusing mode around $35$ and $40$ dB, respectively when compared to the only LoS communication. If the beamforming case is selected, the amount of the increase becomes about $6$ and $9$ dB, for 2 and 4 RIS units respectively.  It is also observed that even numbers of RIS units sustain the best-received power along the way. This issue stems from the uneven positions of the odd numbers of RIS units with respect to the middle point of the road. Since the best results are achieved when the RIS units are located around the BSs, odd numbers of RIS units are obliged to be located asymmetrically. 

The effect of the size of the RIS units on the  RIS modes is also studied. To this aim, 2 RIS units with various sizes are considered and their locations are optimized. The received power vs. the side length and the vehicle location are presented in Fig. \ref{fig:f2} for the beamforming mode, in Fig. \ref{fig:f3} for the focusing mode and, for comparison, in Fig. \ref{fig:f1} for LoS. As seen in (\ref{opti2}), the increase in the RIS sizes in focusing mode leads to higher amounts of the received power. Of course, this increase is limited by the physical law which states that the received power cannot be higher than the transmitted one. On the other hand, in the beamforming mode, the increase in RIS sizes more than $L_{min} = \sqrt{A_{min}}$ does not have any effect on the received power performance. However, the increased numbers of beamformer RIS units can increase the received power since they can lead the signals towards different directions. This situation is not valid for focusing via RIS units, as the performance depends on the area. 

\vspace{-0.2cm}
\section{Conclusion}

In this study, a solution to improve the reliability of the autonomous vehicular networks is proposed via real-time software-controlled reconfigurable intelligent surfaces (RISs).
The optimum locations of the RIS units are determined by solving the formulated optimization problem, considering both the beamforming and focusing operating modes. Significant SNR gains are observed by the optimal deployment of RIS units.

\vspace{-0.2cm}
\section*{Acknowledgment}

The authors would like to thank Ozge Cepheli for the valuable discussions.

\ifCLASSOPTIONcaptionsoff
  \newpage
\fi

\vspace{-0.2cm}

\bibliographystyle{IEEEtran}
\bibliography{main}

\begin{thebibliography}{10}
\providecommand{\url}[1]{#1}
\csname url@samestyle\endcsname
\providecommand{\newblock}{\relax}
\providecommand{\bibinfo}[2]{#2}
\providecommand{\BIBentrySTDinterwordspacing}{\spaceskip=0pt\relax}
\providecommand{\BIBentryALTinterwordstretchfactor}{4}
\providecommand{\BIBentryALTinterwordspacing}{\spaceskip=\fontdimen2\font plus
\BIBentryALTinterwordstretchfactor\fontdimen3\font minus
  \fontdimen4\font\relax}
\providecommand{\BIBforeignlanguage}[2]{{%
\expandafter\ifx\csname l@#1\endcsname\relax
\typeout{** WARNING: IEEEtran.bst: No hyphenation pattern has been}%
\typeout{** loaded for the language `#1'. Using the pattern for}%
\typeout{** the default language instead.}%
\else
\language=\csname l@#1\endcsname
\fi
#2}}
\providecommand{\BIBdecl}{\relax}
\BIBdecl

\bibitem{3gpp}
3GPP, ``{Study on New Radio (NR) access technology},'' {3rd Generation
  Partnership Project (3GPP)}, Technical Specification (TS) 38.912, Jun. 2018,
  version 15.0.0.

\bibitem{6g}
T.~S. {Rappaport}, Y.~{Xing}, O.~{Kanhere}, S.~{Ju}, A.~{Madanayake},
  S.~{Mandal}, A.~{Alkhateeb}, and G.~C. {Trichopoulos}, ``Wireless
  communications and applications above 100 {GHz}: Opportunities and challenges
  for {6G} and beyond,'' \emph{IEEE Access}, vol.~7, pp. 78\,729--78\,757, Jun.
  2019.

\bibitem{uc}
W.~{Yan}, X.~{Yuan}, and X.~{Kuai}, ``Passive beamforming and information
  transfer via large intelligent surface,'' \emph{IEEE Wireless Communications
  Letters}, pp. 1--1, Dec. 2019.

\bibitem{beam}
Q.~{Wu} and R.~{Zhang}, ``Intelligent reflecting surface enhanced wireless
  network via joint active and passive beamforming,'' \emph{IEEE Transactions
  on Wireless Communications}, vol.~18, no.~11, pp. 5394--5409, Nov. 2019.

\bibitem{sec}
M.~{Cui}, G.~{Zhang}, and R.~{Zhang}, ``Secure wireless communication via
  intelligent reflecting surface,'' \emph{IEEE Wireless Communications
  Letters}, vol.~8, no.~5, pp. 1410--1414, Oct. 2019.

\bibitem{sec2}
H.~{Shen}, W.~{Xu}, S.~{Gong}, Z.~{He}, and C.~{Zhao}, ``Secrecy rate
  maximization for intelligent reflecting surface assisted multi-antenna
  communications,'' \emph{IEEE Communications Letters}, vol.~23, no.~9, pp.
  1488--1492, Sep. 2019.

\bibitem{ozgecan}
O.~{Ozdogan}, E.~{Bjornson}, and E.~G. {Larsson}, ``Intelligent reflecting
  surfaces: Physics, propagation, and pathloss modeling,'' \emph{IEEE Wireless
  Communications Letters}, pp. 1--1, Dec. 2019.

\bibitem{nearfar}
S.~W. Ellingson, ``Path loss in reconfigurable intelligent surface-enabled
  channels,'' \emph{arXiv preprint arXiv:1912.06759}, Dec. 2019.

\bibitem{ertugrul}
E.~{Basar}, M.~{Di Renzo}, J.~{De Rosny}, M.~{Debbah}, M.~{Alouini}, and
  R.~{Zhang}, ``Wireless communications through reconfigurable intelligent
  surfaces,'' \emph{IEEE Access}, vol.~7, pp. 116\,753--116\,773, Aug. 2019.

\bibitem{goldsmith}
A.~Goldsmith, \emph{Wireless Communications}.\hskip 1em plus 0.5em minus
  0.4em\relax Cambridge Univ. Press, 2005.

\bibitem{balanis}
C.~Balanis, \emph{Advanced Engineering Electromagnetics, 2nd Edition}.\hskip
  1em plus 0.5em minus 0.4em\relax Wiley, 2012.

\bibitem{experimental}
W.~Tang, M.~Z. Chen, X.~Chen, J.~Y. Dai, Y.~Han, M.~Di~Renzo, Y.~Zeng, S.~Jin,
  Q.~Cheng, and T.~J. Cui, ``Wireless communications with reconfigurable
  intelligent surface: Path loss modeling and experimental measurement,''
  \emph{arXiv preprint arXiv:1911.05326}, Nov. 2019.

\end{thebibliography}

\end{document}